\documentclass[AMA,STIX1COL]{WileyNJD-v2}
\usepackage[utf8]{inputenc}
\usepackage{subfigure}
\usepackage{amsthm}
\usepackage[noabbrev]{cleveref}
\usepackage{autonum}

\allowdisplaybreaks

\newcommand{\orcid}[1]{\href{https://orcid.org/#1}{{\includegraphics[scale=0.05]{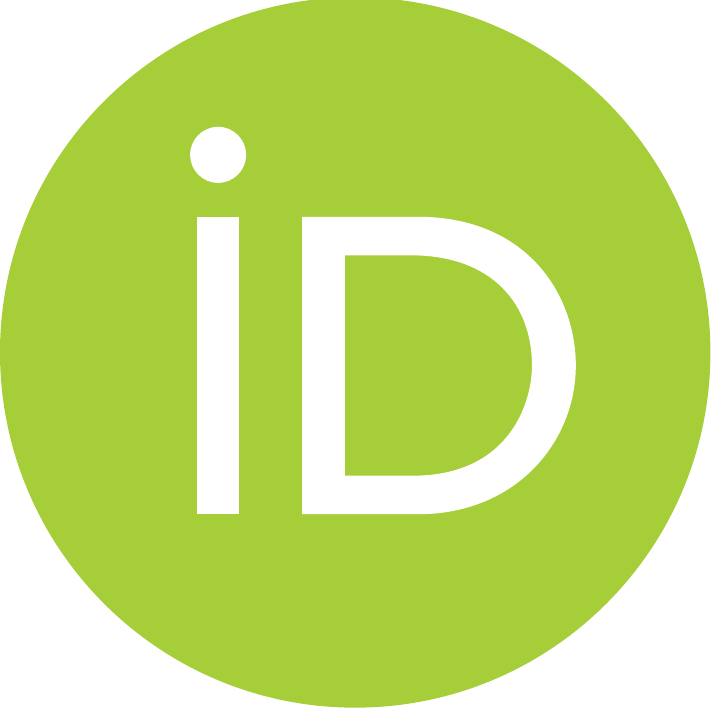}}}}

\newcommand{\ORCID}[2]{\textit{#1} \orcid{#2} \textcolor{blue}{\url{https://orcid.org/#2}}}

\crefname{assumption}{Assumption}{Assumptions}
\theoremstyle{definition}
\newtheorem{example}{Example}

\newcommand{\vph}{\varphi}
\newcommand{\vep}{\varepsilon}
\newcommand{\NH}{\mathcal{H}}
\newcommand{\R}{\mathbb{R}}

\newcommand{\PC}{\mathcal{PC}}
\newcommand{\F}{\widetilde F}
\newcommand{\Kr}{\mathcal{K}_r}
\newcommand{\<}{\leqslant}
\renewcommand{\>}{\geqslant}
\newcommand{\Int}{\int\limits}

\articletype{Research article}

\received{26 April 2016}
\revised{6 June 2016}
\accepted{6 June 2016}

\begin{document}

\title{Necessary and sufficient stability conditions for integral delay systems}%\protect\thanks{This is an example for title footnote.}}

\author[1]{Reynaldo Ortiz}{\orcid{0000-0001-9353-1975}}

\author[2]{Alexey Egorov}{\orcid{0000-0001-7671-2467}}

\author[1]{Sabine Mondié*}{\orcid{0000-0002-0968-1899}}

\authormark{ORTIZ \textsc{et al}}

\address[1]{\orgdiv{Department of Automatic Control}, \orgname{CINVESTAV-IPN}, \orgaddress{\state{Mexico City}, \country{Mexico}}}

\address[2]{\orgdiv{Department of Control Theory}, \\ \orgname{St. Petersburg State University}, \\ \orgaddress{\state{St. Petersburg}, \country{Russia}}}

\corres{*Sabine Mondié, Department of Automatic Control, CINVESTAV-IPN, Av. Instituto Politécnico Nacional 2508, Mexico City 07360, Mexico. \\ \email{smondie@ctrl.cinvestav.mx}}

\presentaddress{Sabine Mondié, Department of Automatic Control, CINVESTAV-IPN, Av. Instituto Politécnico Nacional 2508, Mexico City 07360, Mexico}

\fundingInfo{Grants Conacyt A1-S-24796, SEP-CINVESTAV 155, Russian Foundation for Basic Research 19-01-00146}

\abstract[Summary]{A Lyapunov-Krasovskii functional with prescribed derivative whose construction does not require the stability of the system is introduced. It leads to the presentation of stability/instability theorems. By evaluating the functional at initial conditions depending on the fundamental matrix we are able to present necessary and sufficient stability conditions expressed exclusively in terms of the delay Lyapunov matrix for integral delay systems. Some examples illustrate and validate the stability conditions.}


\jnlcitation{\cname{%
\author{R. Ortiz},
\author{A. Egorov}, and
\author{S. Mondié}} (\cyear{2021}), 
\ctitle{Necessary and sufficient stability conditions for integral delay systems}, \cjournal{Int J Robust Nonlinear Control}, \cvol{2021;1--20}.}

\maketitle

% \footnotetext{\textbf{Abbreviations:} ODE, ordinary differential equation; IDE, integral delay equation}

%%%%%%%%%%%%%%%%%%%%%%%%%%%%%%%%%%%%%%%%%%%%%%%%%%%%%%%%%%%%%%%%%
\section{Introduction} \label{secIntro}

% Integral delay equations arise in the stability analysis and control of differential delay systems one of these is the well-known problem of finite spectrum assignment \cite{ManitusOlbrot1979} where an integral control law allows placing a finite number of poles of the closed-loop system anywhere in the left side of the complex plane, here the internal stability of the control law is one condition for its correct implementation by using numerical methods \cite{Michiels2004}. Non-linear IDEs can be used to describe the spread of periodic seasonal infectious diseases \cite{Cooke1976,Torrejon1976} as measles, mumps and chickenpox \cite{London1973}, in those models the delay represents the length of time an individual remains infectious. An advantage of modelling diseases with IDEs is the possibility to reduce a set of $n$ equations representing $n$ classes to a model of just two equations that represent the susceptible and the infected classes \cite{Smith1978}.

The integral delay equation (IDE) is a particular case of the renewal equation (when the kernel of the integral has a finite support), which has been introduced by Euler~\cite{Euler1760} in his studies on population dynamics and revisited by Lotka~\cite{Lotka1907}. It was adapted to the description of infectious diseases spread by Kermack and McHendrick~\cite{KermackMcKendrick1927}, leading to more complex Erlang-SEIR ordinary differential models with substages. The differential form, describing the number of individuals in a given class, and the integral form, focusing at the spread through time of infection by cohorts of infectors, are shown to be equivalent~\cite{Dushoff2018Equivalence}. The parameters characterizing the ordinary differential equations, in particular the basic reproductive number $R_0$ and the generation-interval distribution are connected through the Euler-Lotka equation. While the differential form is preferred for the evaluation of optimal mitigation strategies, epidemiologists agree that the generation-interval distribution is easier to infer~\cite{DushoffInfering2020} from contact tracing data at the early exponential growth stage of an outbreak. Moreover, the IDE description may have advantages in numerical simulations as it reduces compartmental descriptions to a single integral equation. It is not our purpose to review here the outstanding body of theoretical and practical studies on biological applications of integral delay systems~\cite{BellmanCooke,Cooke1976,Torrejon1976,London1973,ArinoDriessche2006,Breda2012}, but to point out its high relevance in the context of the current Covid-19 pandemics, as a motivation for our study of stability. 

IDEs also arise in the challenging and widespread engineering control problem of systems with input delays: in the well-known finite spectrum assignment~\cite{ManitusOlbrot1979} strategy, an integral control law allows, under a spectral controllability assumption, placing a closed-loop system finite number of poles anywhere in the left-hand side of the complex plane. Here, the stability of an IDE describing internal dynamics is shown to be crucial for the correct implementation of the control law when using numerical methods~\cite{Michiels2004}. 

In a more theoretical vein, the analysis of the stability of IDEs has received attention recently, in particular in the Lyapunov-Krasovskii/LMI's framework that provides sufficient stability conditions~\cite{M-Aguilar2010,MondieMelchor2012,Damak2014,Damak2015} and robust stability results~\cite{Melchor2016}. The aim of the present contribution is to present necessary and sufficient stability conditions for a class of IDE. The Lyapunov functionals and matrices framework presented in the book authored by Kharitonov~\cite{Kharitonov2013} is the privileged one in this pursuit, as attested by the past decade necessary and sufficient stability results on pointwise~\cite{Egorov2014}, distributed~\cite{EgorovCuvas2017}, and neutral type differential delays systems~\cite{Gomez2018}. 
The basic theory for integral equation is not covered in the above mentioned monograph, but in the case of exponentially stable systems, a functional with prescribed derivative is available~\cite{DMelchor2010} and a new definition of the Lyapunov matrix lead to necessary stability conditions expressed in terms of the Lyapunov matrix~\cite{Ortiz2019TAC}. Notice that the presented framework carries important issues on the Lyapunov matrices as the existence and uniqueness~\cite{BellmanCooke,Ortiz2019,Arismendi2019}, and of course its computation~\cite{Arismendi2017,Ortiz2018}.

The paper is organized as follows: The integral delay system is introduced in \Cref{secSystem}, some auxiliary elements are given in \Cref{secAuxEl}, new properties of the Lyapunov matrix proved without stability assumption, a crucial issue when looking for necessary and sufficient conditions, and a new Lyapunov-Krasovskii functional, which satisfies a prescribed derivative independently of the system stability are introduced in \Cref{secLyapMatr}. \Cref{secKeyResults} is devoted to key results required in the main proof; stability/instability results for a modified functional and a bilinear functional. We arrive at the main result in \Cref{secMain}: necessary and sufficient stability conditions expressed exclusively in terms of the delay Lyapunov matrix of the IDE. We validate the necessary and sufficient stability conditions by examples in \Cref{secExamples} and end the paper with some concluding remarks. For the sake of readability, most of the proofs are given in the appendix.

\textit{Notation:} The smallest eigenvalue of a symmetric matrix $Q$ is denoted by $\lambda_{\min}(Q)$, while notation $Q>0$, $Q \ngeqslant 0$ means that $Q$ is positive definite and not positive semi-definite, respectively. The block matrix of $r\times r$ blocks with the block $A_{ij}$ in the $i$-th row and $j$-th column is written as $\left[A_{ij}\right]_{i,j=1}^{r}$. The euclidean norm for vectors is denoted by $\|\cdot\|$, $\PC_h$ is the space of piecewise continuous bounded functions of dimension $n$ defined on $[-h,0)$, the weak derivative of a function $u$ with respect to its argument $t$ is written as $\dfrac{d}{dt}u$ or $u'$, if exists.

%%%%%%%%%%%%%%%%%%%%%%%%%%%%%%%%%%%%%%%%%%%%%%%%%%%%%%%%%%%%%%%%%
\section{The System} \label{secSystem}

Consider the linear integral delay system
\begin{equation} \label{eqSys}
x(t)=\int_{-h}^{0}F(\theta)x(t+\theta)d\theta, \quad t\>0,
\end{equation}
where $x(t)\in\R^n$, delay $h>0$. In this paper, function $x(t,\vph)$, $t\>-h$, is a solution of system~\eqref{eqSys}, corresponding to the initial function $\vph\in\PC_h$, if it is a piecewise-continuous function, defined on $[-h,\infty)$, satisfying the initial condition
\[ x(\theta,\vph)=\vph(\theta), \quad \theta\in[-h,0). \]

Kernel $F$ in formula~\eqref{eqSys} is assumed to be a function of bounded variation. Notice that without any loss of generality we can assume that at every point of segment $[-h,0]$ function $F$ is either left or right continuous.

Reduce now the initial value problem for system~\eqref{eqSys} to the renewal equation. If we introduce a simple extension
\begin{equation} \label{eqFtilde}
\F(t)=
\begin{cases}
F(t),& t\in[-h,0], \\
0, & \text{elsewhere},
\end{cases}
\end{equation}
of function $F$, we obtain the equation
\begin{equation} \label{eqRen}
x(t)=g(t) +\int_{0}^{t}\F(\theta-t)x(\theta)d\theta,\;\;t\>0,
\end{equation}
where
\begin{equation}
g(t)=
\begin{cases}
\Int_{-h}^{-t}F(\theta)\vph(t+\theta)d\theta,& t\in[0,h], \\
0, & \text{elsewhere},
\end{cases}
\end{equation}
is a continuous on $[0,\infty)$ function.

It has been shown by Bellman and Cooke~\cite{BellmanCooke} that in the scalar case equation~\eqref{eqRen} has a unique solution, and the solution is continuous, if function $g$ is. For the nonscalar case the result can be proved in a similar way.

The restriction of the solution to the interval $[t-h,t)$ is defined by $x_t(\vph)$:
\[ x_t(\vph)(\theta)=x(t+\theta,\vph), \quad\theta\in[-h,0). \]

%%%%%%%%%%%%%%%%%%%%%%%%%%%%%%%%%%%%%%%%%%%%%%%%%%%%%%%%%%%%%%%%%
\section{Auxiliary Elements} \label{secAuxEl}

\begin{definition}
The complex number $s$ is called an \textit{eigenvalue} of system~\eqref{eqSys}, if $\det H(s)=0$, where $H$ is the characteristic matrix:
\begin{equation} \label{eqCharM}
H(s)=I-\int_{-h}^{0}e^{s\theta}F(\theta)d\theta.
\end{equation}
\end{definition}

\begin{assumption} \label{assEig}
System~\eqref{eqSys} does not have pure imaginary eigenvalues.
\end{assumption}

We introduce the seminorm $\|\cdot\|_{\NH}$:
\[ \|\vph\|_\NH =\sqrt{\int_{-h}^{0} \left\|\vph(\theta)\right\|^2 d\theta}. \]

\begin{definition} \label{Def1}
System \eqref{eqSys} is said to be \textit{exponentially stable}, if there exist constants $\sigma>0$ and $\gamma\> 1$, such that for any $\vph\in\PC_h$
\[ \left\|x(t,\vph)\right\| \<\gamma\|\vph\|_\NH e^{-\sigma t}, \quad t\>0. \]
\end{definition}

The $n\times n$ matrix $K(t)$, $t\>-h$, is known as the fundamental matrix of~\eqref{eqSys} and is the unique solution of the equation
\begin{equation} \label{eqK}
K(t)=\int_{-h}^{0} F(\theta) K(t+\theta) d\theta, \quad t\>0,
\end{equation} 
with initial condition,
\begin{equation} \label{eqK0}
K(t)=K_0=\left(\int_{-h}^{0}F(\theta)d\theta-I\right)^{-1}, \quad t<0.
\end{equation}
Under~\Cref{assEig} the matrix in brackets is invertible. Notice that the matrix defined in~{\eqref{eqK}-\eqref{eqK0}} satisfies also the equation
\begin{equation} \label{eqKK}
K(t)=\int_{-h}^{0}K(t+\theta)F(\theta)d\theta, \quad t\>0,
\end{equation}
that can be proved via the Laplace transform, like it has been done for differential-difference systems~\cite{BellmanCooke}.

\begin{lemma} \label{thmLK0K0I}
The fundamental matrix $K$ is such that
\begin{equation} \label{eqK0K0I}
K(0)-K_0=I.
\end{equation}
\end{lemma}
\begin{proof}
We evaluate equality~\eqref{eqK} (resp. equality~\eqref{eqKK}) at $t=0$, as $\theta\in[-h,0]$, then $K(\theta)=K_0$ almost everywhere. Subtracting and factorizing $K_0$ on the right (resp. on the left), \Cref{thmLK0K0I} follows by~\eqref{eqK0}.
\end{proof}

\begin{lemma} \label{thmKprime}
The fundamental matrix $K$ is absolutely continuous on $[0,\infty)$ and its weak derivative satisfies equation
\begin{equation} \label{eqWeakDerEq}
K'(t)= \F(-t)+\int_0^t K'(\theta)\F(\theta-t)d\theta,\;\; t\>0.
\end{equation}
\end{lemma}
\begin{proof}
The renewal equation~\eqref{eqWeakDerEq} has a unique solution $K'$, which is Lebesgue integrable and bounded~\cite{BellmanCooke}. It remains to show that the solution is a weak derivative of the fundamental matrix.

As the fundamental matrix is unique, it remains to show that the function
\begin{equation} \label{eqKasAC}
K(t)=
\begin{cases}
K(0)+\Int_0^t K'(s)ds, & t\>0, \\
K_0, & t<0,
\end{cases}
\end{equation}
satisfies equality~\eqref{eqKK}.

Notice that the equality can be rewritten in the form
\begin{equation} \label{eqKeqToSubst}
K(t)= K_0\int_{-h}^{-t}\F(\theta)d\theta +\int_0^t K(\theta)\F(\theta-t)d\theta,\;\; t\>0.
\end{equation}
To finish the proof it remains to substitute~\eqref{eqKasAC} into~\eqref{eqKeqToSubst}, and obtain an identity, using equality~\eqref{eqWeakDerEq}.
\end{proof}

\begin{corollary}
For exponentially stable systems, functions $\|K(t)\|$ and $\|K'(t)\|$, $t>0$, exponentially decrease to zero while $t$ increases.
\end{corollary}

Notice that by changing the order of integration the Cauchy formula~\cite{Ortiz2019TAC} can be rewritten as
\begin{equation} \label{eqCauchy1}
x(t,\vph)=\int_{-h}^{0}\dfrac{d}{dt}\int_{-h}^{\xi}K(t-\xi+\theta)F(\theta)d\theta \vph(\xi)d\xi, \quad t\>0.
\end{equation}

%%%%%%%%%%%%%%%%%%%%%%%%%%%%%%%%%%%%%%%%%%%%%%%%%%%%%%%%%%%%%%%%%
\section{The Lyapunov Matrix and the Complete Type Functional} \label{secLyapMatr}

\subsection{The Case of Exponential Stability} \label{secLMexpst}

Introduce the Lyapunov matrix for equation~\eqref{eqSys}.

\begin{lemma}\cite{Ortiz2019TAC}
Let~\eqref{eqSys} be exponentially stable. For every $n \times n$ matrix $W$ the matrix
\begin{equation} \label{eqU}
U(\tau)=\int_{0}^{\infty}\left(K(t)-K_0\right)^T WK(t+\tau)dt
\end{equation}
is well defined for $\tau\in\R$.
\end{lemma}

The matrix-valued function $U$ is called the Lyapunov matrix for system~\eqref{eqSys}. This matrix is the core element of the necessary and sufficient stability conditions in the main result of this paper. For an exponentially stable system~\eqref{eqSys}, the Lyapunov matrix~\eqref{eqU} associated with a symmetric matrix $W$ satisfies the dynamic property
\begin{equation} \label{eqDyn}
U(\tau)=\int_{-h}^{0}U(\tau+\theta)F(\theta)d\theta, \quad \tau\> 0,
\end{equation}
and the symmetry property
\begin{equation} \label{eqSym}
U(-\tau) =U^T(\tau) +P-\tau K_0^TWK_0, \quad\tau\in\R,
\end{equation}
where the skew-symmetric matrix $P$ is given by
\begin{gather}
P=S^T WK_0 -K_0^TW S, \label{eqP} \\
S=K_0\int_{-h}^0\theta F(\theta)d\theta K_0. \label{eqS}
\end{gather}

We will need the following technical result.

\begin{lemma} \label{thmUprimeprime}
If system~\eqref{eqSys} is exponentially stable, for $\tau<0$
\begin{equation} \label{eqUprimeprime}
U''(\tau) =-K'^T(-\tau)W -\int_{0}^{\infty}K'^T(t-\tau)WK'(t)dt.
\end{equation}
\end{lemma}
\begin{proof}
If $\tau<0$,
\begin{equation}
U(\tau)=\int_{0}^{\infty}\left(K(t)-K_0\right)^TWK(t+\tau)dt =\left(\int_{0}^{-\tau} K^T(t)dt+\tau K_0^T\right)WK_0 +\int_{-\tau}^{\infty}\left(K(t)-K_0\right)^T WK(t+\tau)dt.
\end{equation}
The first derivative is equal to
\begin{multline}
U'(\tau) =-K^T(-\tau)WK_0 +K_0^TWK_0 +\left(K(-\tau)-K_0\right)^TWK(0) +\int_{-\tau}^{\infty}\left(K(t)-K_0\right)^TWK'(t+\tau)dt\\
=K^T(-\tau)W-K_0^TW +\int_{0}^{\infty}\left(K(t-\tau)-K_0\right)^TWK'(t)dt.
\end{multline}
Therefore, the second derivative is given by formula~\eqref{eqUprimeprime}.
\end{proof}

Based on the Lyapunov matrix we define the quadratic functional $v_0(\vph)$, $\vph\in\PC_h$, that satisfies the equality
\begin{equation} \label{eqDv0}
\dfrac{d}{dt}v_0(x_t(\vph)) =-x^T(t,\vph) W x(t,\vph), \quad t\>0,
\end{equation}
along the trajectories of~\eqref{eqSys} for a given symmetric positive definite matrix $W$.

Integrating~\eqref{eqDv0} from $t=0$ to $t=\mathcal{T}>0$ and the assumption that~\eqref{eqSys} is exponentially stable, $x_{\mathcal{T}}\rightarrow 0$ as $\mathcal{T}\rightarrow\infty$, implies that
\begin{equation} \label{eqv0int}
v_0(\vph)=\int_{0}^{\infty}x^T(t,\vph) W x(t,\vph)dt.
\end{equation}
Substitution of the Cauchy formula~\eqref{eqCauchy1} into~\eqref{eqv0int} yields
\begin{equation}
v_0(\vph)=\int_{0}^{\infty}\left(\int_{-h}^{0}\vph^T(\xi_1)\dfrac{d}{dt} \int_{-h}^{\xi_1}F^T(\theta_1)K^T(t-\xi_1+\theta_1)d\theta_1 d\xi_1\right) W \left(\int_{-h}^{0}\dfrac{d}{dt} \int_{-h}^{\xi_2} K(t-\xi_2+\theta_2) F(\theta_2)d\theta_2 \vph(\xi_2)d\xi_2 \right)dt.
\end{equation}
Changing the order of integration, we obtain
\begin{equation} \label{eqv0Q}
v_0(\vph)=\int_{-h}^{0}\int_{-h}^{0}\vph^T(\xi_1)Q(\xi_1,\xi_2)\vph(\xi_2)d\xi_2 d\xi_1,
\end{equation}
where, for $\xi_1,\,\xi_2\in[-h,0]$,
\begin{equation} \label{eqQIni}
Q(\xi_1,\xi_2)=\int_{0}^{\infty}\left(\dfrac{d}{dt}\int_{-h}^{\xi_1}F^T(\theta_1) K^T(t-\xi_1+\theta_1) d\theta_1\right) W\left( \dfrac{d}{dt}\int_{-h}^{\xi_2}K(t-\xi_2+\theta_2)F(\theta_2) d\theta_2\right)dt.
\end{equation}

The proof of the following theorem is given in~\Cref{appQ}.

\begin{theorem} \label{thmQ}
Let system~\eqref{eqSys} be exponentially stable. For $\xi_1,\,\xi_2\in[-h,0]$
\begin{equation} \label{eqQ}
Q(\xi_1,\xi_2) =-\int_{-h}^{\xi_2}U''(\xi_1-\xi_2+\theta)F(\theta)d\theta +\int_{\xi_1}^{0}\int_{-h}^{\xi_2} F^T(\theta_1)U''(\xi_1-\theta_1-\xi_2+\theta_2) F(\theta_2)d\theta_2d\theta_1.
\end{equation}
\end{theorem}

%%%%%%%%%%%%%%%%%%%%%%%%%%%%%%%%%%%%%%%%%%%%%%%%%%%%%%%%%%%%%%%%%
\subsection{The General Case} \label{secLMGeneralCase}

Properties~\eqref{eqDyn} and~\eqref{eqSym} allow introducing a new definition of the Lyapunov matrix, which does not require the exponential stability of system~\eqref{eqSys}.

\begin{definition} \label{defLyap}
Let $W$ be a positive definite matrix. Function $U(\tau)\in\R^{n\times n}$, $\tau\in\R$, which is absolutely continuous and has bounded first and second weak derivatives on any segment $[a,b]$, such that $0\not\in(a,b)$ (i.\:e., the set does not contain zero as its interior point), is called \textit{the Lyapunov matrix} for system~\eqref{eqSys}, if it satisfies the dynamic property~\eqref{eqDyn} and the symmetry property~\eqref{eqSym}.
\end{definition}

\begin{remark}
The problem of the existence of the Lyapunov matrix arises. It has been shown in~\cite{Ortiz2019} that for a special case of system~\eqref{eqSys} the Lyapunov matrix exists and is unique if and only if the Lyapunov condition holds, i.\:e., the system has no eigenvalues, which are symmetric with respect to zero. This result is expected to be also true in the general case, but requires a separate investigation.
\end{remark}

\begin{remark}
It is worth mentioning that as shown in~\Cref{secLMexpst}, in the case of exponential stability matrix~\eqref{eqU} satisfies the new definition. 
\end{remark}

We present new properties that do not rely on the system stability, a feature that is crucial in the proof of sufficient stability conditions. The proof of the following lemma is in~\Cref{appLMprop1}.

\begin{lemma} \label{thmGDP}
The Lyapunov matrix satisfies the following properties for $\tau\in\R$:
\begin{equation} \label{eqAlgNeg}
U(\tau) =\int_{-h}^{0}F^T(\theta)U(\tau-\theta)d\theta +WS -W\int_0^{\tau}K(s)ds,
\end{equation}
\begin{equation} \label{eqAlg}
U(\tau)=\int_{-h}^{0}U(\tau+\theta)F(\theta)d\theta +\int_{-\tau}^0\left(K(s)-K_0\right)^T ds W.
\end{equation}
\end{lemma}

\begin{corollary} \label{thmU'}
For $\tau\in\R$
\begin{equation} \label{eqAlgNegDer}
U'(\tau)=\int_{-h}^{0}F^T(\theta)U'(\tau-\theta)d\theta -WK(\tau).
\end{equation}
\end{corollary}

The next result can be easily derived from formula~\eqref{eqAlg}, taking into account that for negative $\tau$, the argument of matrix $U$ in the formula is always negative.

\begin{corollary} \label{thmU''}
If $\tau<0$,
\[ U''(\tau)=\int_{-h}^0 U''(\tau+\theta)F(\theta)d\theta -K'^T(-\tau)W. \]
\end{corollary}

The proof of the next lemma is given in~\Cref{appLMprop2}.

\begin{lemma} \label{thmT1T2}
The Lyapunov matrix satisfies the following relation with the fundamental matrix $K$ for $\tau_1,\tau_2\in[0,h]$
\begin{equation}
U(\tau_2-\tau_1)=\int_{-h}^{0}\int_{-h}^{\xi} U'(-\tau_1-\xi+\theta) F(\theta) d\theta K(\tau_2+\xi) d\xi +\int_{-h}^{0}\left(K(\tau_1+\xi)-K_0\right)^TWK(\tau_2+\xi)d\xi.
\end{equation}
\end{lemma}

Consider functional~\eqref{eqv0Q}. It is important to notice that formula~\eqref{eqQ}, in contrast to formula~\eqref{eqQIni}, makes sense even if the system is not exponentially stable, and only the existence of the Lyapunov matrix is needed. Thus, we can use formula~\eqref{eqQ} as a new definition of function $Q$. If we substitute this function into~\eqref{eqv0Q}, we can show that the functional satisfies equality~\eqref{eqDv0} independently of the stability of system~\eqref{eqSys}. The proof can be found in~\Cref{appDiffV0}.

\begin{theorem} \label{thmDv0}
Functional~\eqref{eqv0Q} with $Q$ defined by formula~\eqref{eqQ} satisfies equality~\eqref{eqDv0} along the trajectories of system~\eqref{eqSys}.
\end{theorem}

%%%%%%%%%%%%%%%%%%%%%%%%%%%%%%%%%%%%%%%%%%%%%%%%%%%%%%%%%%%%%%%%%
\section{Key Results} \label{secKeyResults}

In this section we introduce the new functional $v_1$ and the corresponding bilinear functional, and present fundamental stability/instability results.

Consider the bilinear functional
\begin{equation} \label{eqZ}
z(\vph,\psi) =\int_{-h}^{0}\int_{-h}^{0}\vph^T(\xi_1)Q(\xi_1,\xi_2)\psi(\xi_2)d\xi_2 d\xi_1 + \int_{-h}^{0}\vph^T(\xi)W\psi(\xi)d\xi,\;\;\vph,\psi\in\PC_h,
\end{equation}
with corresponding quadratic functional
\begin{equation} \label{eqv1}
v_1(\vph)=z(\vph,\vph) =v_0(\vph)+\int_{-h}^{0} \vph^T(\theta)W\vph(\theta)d\theta, \;\;\vph\in\PC_h.
\end{equation}

\begin{theorem} \label{thmV1prop}
Functional $v_1$ is such that the following hold:
\begin{enumerate}[i)]
\item Its derivative along the trajectories of~\eqref{eqSys} is equal to
\begin{equation} \label{eqDv1}
\dfrac{d}{dt}v_1(x_t(\vph))=-x^T(t-h,\vph)Wx(t-h,\vph).
\end{equation}
\item It is continuous in the sense that for any $\vph\in\PC_h$ and any $\vep>0$ there exists $\delta>0$ such that
\[ \psi\in\PC_h,\;\int_{-h}^0 \|\vph(\theta)-\psi(\theta)\|\,d\theta <\delta \;\;\Longrightarrow \;\; |v_1(\vph)-v_1(\psi)|<\vep.\]
\item If system~\eqref{eqSys} is exponentially stable, then~\eqref{eqv1} admits a lower bound
\begin{equation} \label{eqLB}
v_1(\vph)\>\alpha\|\vph\|_\NH^2,\quad\alpha>0.
\end{equation}
\end{enumerate}
\end{theorem}
\begin{proof}
The derivative of functional $v_0$ is equal to~\eqref{eqDv0}. Applying the change of variable $s=t+\theta$, the derivative of the second term in the r.h.s of~\eqref{eqv1} is
\[ \dfrac{d}{dt}\int_{t-h}^{t}x^T(s,\vph)W x(s,\vph) ds =x^T(t,\vph)Wx(t,\vph) -x^T(t-h,\vph)Wx(t-h,\vph), \]
and i) follows directly.

The continuity of the functional is a consequence of the uniform boundedness of function $Q$ defined by formula~\eqref{eqQ}. Function $Q$ is bounded, as the second derivative of the Lyapunov matrix is (by \Cref{defLyap}).

To prove iii), we define the functional
\[ \bar{v}(\vph) =v_1(\vph) -\alpha\|\vph\|_{\NH}^2 =v_1(\vph) -\alpha\int_{-h}^{0}\|\vph(\theta)\|^2 d\theta. \]
Then
\[ \dfrac{d}{dt}\bar{v}(x_t(\vph))=-\bar{w}(x_t(\vph)) =-x^T(t-h,\vph)W x(t-h,\vph) -\alpha\left(x^T(t,\vph)x(t,\vph) -x^T(t-h,\vph)x(t-h,\vph)\right). \]
If $\alpha\in[0,\lambda_{\min}(W)]$, then $\bar w(x_t(\vph))\>0$. As system~\eqref{eqSys} is exponentially stable, we can represent $\bar{v}$ as
\[ \bar{v}(\vph)=\int_{0}^{\infty}\bar{w}(x_t(\vph))dt\>0, \]
and iii) follows.
\end{proof}

Consider the functions of the form
\begin{equation} \label{eqSpec}
\psi(\theta) =\sum_{i=1}^{r} \left(K(\tau_i+\theta)-K_0\right)\gamma_i,
\end{equation}
where $r$ is a positive integer, 
\[ \tau_i\in(0,h],\; \gamma_i\in\R^n,\;\; i=1,\ldots,r. \]

Now we prove that when valued at functions of the form~\eqref{eqSpec} our functional $v_1$ takes a very simple quadratic form, which is based on a finite number of values of the Lyapunov matrix, and does not contain any integrals, which are hard to compute. The proof is given in \Cref{appCIZ}.

\begin{lemma} \label{thmCIZ}
For any $\tau_1,\tau_2\in(0,h]$, arbitrary vectors $\gamma_1,\gamma_2\in\R^n$ and functions
\begin{equation} \label{eqIC}
\vph_i(\theta)=\left(K(\tau_i+\theta) -K_0\right)\gamma_i, \quad \theta\in[-h,0), \;\; i=1,2,
\end{equation}
the bilinear functional $z$ can be expressed as
\begin{equation} \label{eqZL}
z(\vph_1,\vph_2) =\gamma_1^TL(\tau_1,\tau_2)\gamma_2,
\end{equation}
where
\begin{equation} \label{eqL}
L(\tau_1,\tau_2) =U(0)-U(-\tau_1) -U(\tau_2) +U(\tau_2-\tau_1).
\end{equation}
\end{lemma}

\begin{corollary} \label{thmV1onSpec}
For any function $\psi$ of the form~\eqref{eqSpec}
\begin{equation} \label{eqVL}
v_1(\psi) =\gamma^T \Kr(\tau_1,\ldots,\tau_r)\gamma,
\end{equation}
where the block-matrix
\begin{equation} \label{eqMK}
\Kr(\tau_1,\ldots,\tau_r) =\left[L(\tau_i,\tau_j)\right]_{i,j=1}^{r},
\end{equation}
and vector-column
\[ \gamma=\left[\gamma_1^T,\ldots,\gamma_r^T\right]^T. \]
\end{corollary}

\begin{lemma}
The block matrix $\Kr(\tau_1,\ldots,\tau_r)$ is symmetric, i.\:e., the following equality is satisfied
\begin{equation} \label{eqSymL}
L(\tau_i,\tau_j)=L^T(\tau_j,\tau_i).
\end{equation} 
\end{lemma}
\begin{proof}
By definition~\eqref{eqL}, we have
\[ L(\tau_j,\tau_i)=U(0)-U(-\tau_j) -U(\tau_i)+U(\tau_i-\tau_j), \]
with the help of the symmetric property~\eqref{eqSym}, we can write
\[ L(\tau_j,\tau_i)=U^T(0)-U^T(\tau_j) -U^T(-\tau_i)+U^T(\tau_j-\tau_i), \]
which implies~\eqref{eqSymL}.
\end{proof}

Notice that in the case of equidistant points 
\begin{equation} \label{eqEquidTau}
\tau_i=i\dfrac{h}{r}, \quad i=1,\ldots,r,
\end{equation}
the matrix $\Kr$ is of the form
\begin{equation} \label{eqKrE}
\Kr\left(\dfrac{h}{r},2\dfrac{h}{r},3\dfrac{h}{r},\ldots,h\right) =\left[L\left(i\dfrac{h}{r},j\dfrac{h}{r}\right) \right]_{i,j=1}^{r}.
\end{equation}

Now we give an instability theorem. A similar result, which is based on an interesting idea~\cite{MedvZh2013}, was previously proven in the context of differential systems with pointwise~\cite{Egorov2014} and distributed delays~\cite{EgorovCuvas2017}. It states that if system~\eqref{eqSys} is unstable, then the functional $v_1$ is unbounded from below. The proof is given in~\Cref{appUnsTh}.

\begin{theorem} \label{thmUnsTh}
If system~\eqref{eqSys} is unstable and satisfies \Cref{assEig}, then for every $c>0$ there exists a function of the form~\eqref{eqSpec} with equidistant points~\eqref{eqEquidTau}, such that
\begin{equation} \label{eqUnsTh}
v_1(\psi)\<-c.
\end{equation}
\end{theorem}

%%%%%%%%%%%%%%%%%%%%%%%%%%%%%%%%%%%%%%%%%%%%%%%%%%%%%%%%%%%%%%%%%
\section{Main Result: Stability Conditions} \label{secMain}

We are now ready to present the main contribution of this work, necessary and sufficient stability conditions for integral delay system~\eqref{eqSys} in terms of the delay Lyapunov matrix.

\begin{theorem}
Let \Cref{assEig} hold. System~\eqref{eqSys} is exponentially stable if and only if for every natural number $r\>2$
\begin{equation} \label{eqCNS1}
\left[L\left(i\dfrac{h}{r},j\dfrac{h}{r}\right) \right]_{i,j=1}^{r} > 0.
\end{equation}
Moreover, if system~\eqref{eqSys} is unstable, then there exists a natural number $r$, such that
\begin{equation} \label{eqCNS2}
\left[L\left(i\dfrac{h}{r},j\dfrac{h}{r}\right) \right]_{i,j=1}^{r} \ngeqslant 0.
\end{equation}
\end{theorem}
\begin{proof}
\textit{Necessity:} As we assume that system~\eqref{eqSys} is exponentially stable, \Cref{thmV1prop} toguether with \Cref{thmV1onSpec} imply that
\[ v_1(\psi)=\gamma^T \left[L\left(i\dfrac{h}{r},j\dfrac{h}{r}\right) \right]_{i,j=1}^{r}\gamma \> \alpha\|\psi\|_\NH^2, \]
if $\psi$ is defined by~\eqref{eqSpec} with equidistant points~\eqref{eqEquidTau}. It remains to demonstrate that $\|\psi\|_\NH>0$, if $\gamma\neq 0$.

By contradiction we assume that $\gamma\neq 0$ but $\|\psi\|_\NH=0$. Hence, $\psi(\theta)=0$ almost everywhere. But this is impossible, as $K(0)-K_0=I$ and $K(t)$ is right continuous at zero.

\textit{Sufficiency:} Assume that system~\eqref{eqSys} is unstable. We need to prove that there exists a natural number $r$ that satisfies~\eqref{eqCNS2}, which is equivalent to the existence of a vector $\gamma$ that satisfies the inequality: 
\begin{equation} \label{eqSufKr}
\gamma^T\left[L\left(i\dfrac{h}{r},j\dfrac{h}{r}\right) \right]_{i,j=1}^{r}\gamma < 0.
\end{equation}
By \Cref{thmV1onSpec} the left hand side of the inequality is equal to $v_1(\psi)$ for a function $\psi$ of the form~\eqref{eqSpec} with equidistant points~\eqref{eqEquidTau}. By \Cref{thmUnsTh} such function satisfying~\eqref{eqSufKr} exists.
\end{proof}

\begin{remark}
A notable fact is that the array $\Kr$ of Lyapunov matrices is similar to the one obtained for difference equations in continuous time~\cite{Rocha2017}, but differs from the one for differential systems~\cite{EgorovCuvas2017,EgorovMondie2014,Gomez2017}. Of course, the underlying Lyapunov matrices are those corresponding to the class of systems under consideration.
\end{remark}

%%%%%%%%%%%%%%%%%%%%%%%%%%%%%%%%%%%%%%%%%%%%%%%%%%%%%%%%%%%%%%%%%
\section{Illustrative Examples} \label{secExamples}

This section illustrates how we can find the exact stability region in a given space of parameters through the necessary and sufficient stability conditions. We check the stability condition~\eqref{eqCNS1} at each point for increasing values of the parameter $r$, i.\:e.,
\begin{align}
r=2: & \begin{bmatrix}
	L\left(\frac{h}{2},\frac{h}{2}\right) & L\left(\frac{h}{2},h\right) \\
	L\left(h,\frac{h}{2}\right) & L\left(h,h\vphantom{\frac{h}{2}}\right)
\end{bmatrix} > 0, \\
r=3: & \begin{bmatrix}
	L\left(\frac{h}{3},\frac{h}{3}\right) & L\left(\frac{h}{3},\frac{2h}{3}\right) & L\left(\frac{h}{3},h\right) \\
	L\left(\frac{2h}{3},\frac{h}{3}\right) & L\left(\frac{2h}{3},\frac{2h}{3}\right) & L\left(\frac{2h}{3},h\right) \\
	L\left(h,\frac{h}{3}\right) & L\left(h,\frac{2h}{3}\right) & L\left(h,h\vphantom{\frac{h}{2}}\right)
\end{bmatrix} > 0,
\end{align}
and so on.

As the conditions are necessary for any $r$, we obtain outer estimates of the exact stability region. Notice that points that do not fulfill the stability condition for some $r$, can be excluded from tests for greater values of $r$.

In the figures of this section, we test the condition~\eqref{eqCNS1} on a grid of 50 by 50 equidistant points of the space of given parameters. The points that satisfy the positivity condition are depicted by isolated points. The continuous lines are hyper-surfaces corresponding to imaginary axis root crossings computed using the D-subdivisions method.

\begin{example} \label{ex:1}
Consider the problem of finite spectrum assignment of a double integrator with delayed input~\cite{DMelchor2010},
\begin{equation} \label{eqEj1}
\dot{x}(t)=Ax(t)+Bu(t-h),
\end{equation}
where
\begin{equation}
A=\left[\begin{array}{cc}
	0 & 1 \\
	0 & 0
\end{array}\right], \quad B=\left[\begin{array}{c}
	0 \\
	1
\end{array}\right], \quad h=1.
\end{equation}
The control law
\begin{equation} \label{eqEj1u}
u(t)=Cx(t+h)=C\left(e^{Ah}x(t)+\int_{-h}^{0}e^{-A\theta} Bu(t+\theta)d\theta\right),
\end{equation}
where $C=\left[c_1,c_2\right]$, assigns the spectrum of the matrix $A+BC$ to the closed-loop system~{\eqref{eqEj1}-\eqref{eqEj1u}}~\cite{ManitusOlbrot1979}.

The internal dynamics of~\eqref{eqEj1u} is described by the integral equation~\cite{DMelchor2010}
\[ z(t)=\int_{-h}^{0}Ce^{-A\theta}Bz(t+\theta)d\theta = \int_{-h}^{0}(c_2-c_1\theta) z(t+\theta)d\theta. \]

The space of design parameters $(c_1,c_2)$ is depicted in Figure~\ref{FigEj1}. The improvement of the outer estimate of the stability region when $r$ in condition~\eqref{eqCNS1} increases is clear. 
\begin{figure}[!t]
\centering
\subfigure[]{\includegraphics[scale=0.7]{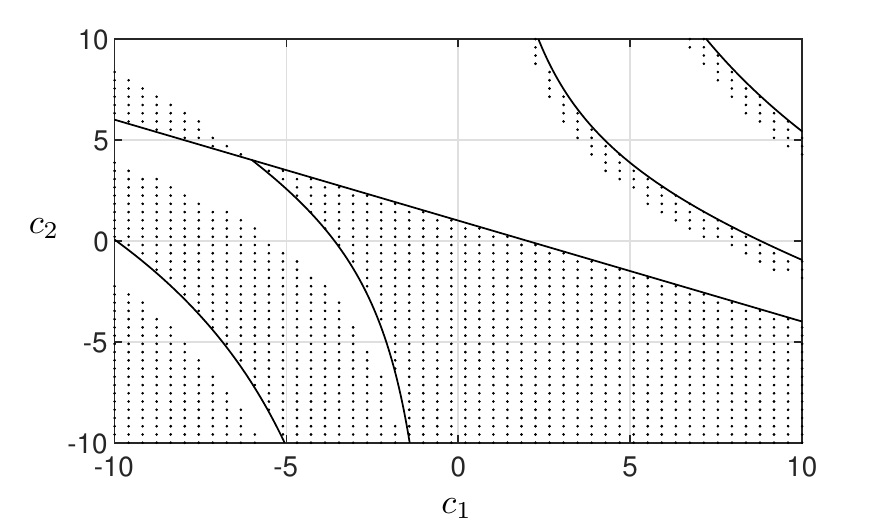}}
\subfigure[]{\includegraphics[scale=0.7]{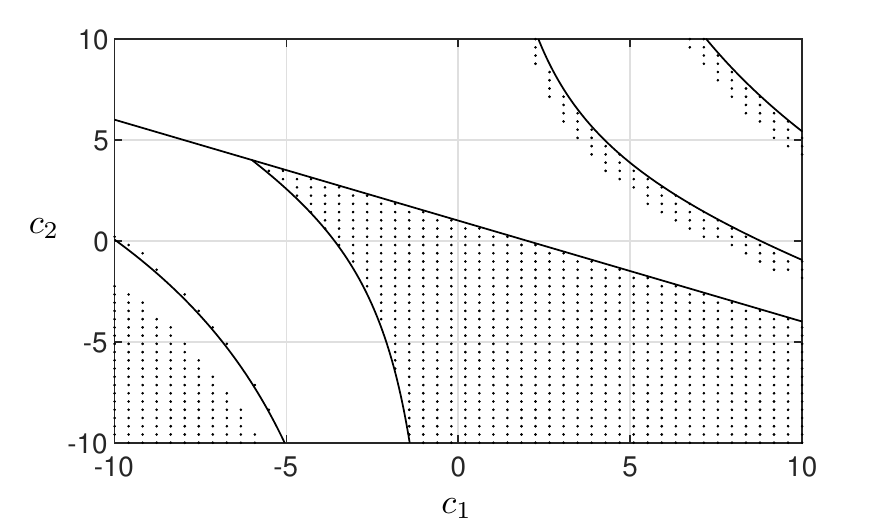}}
\subfigure[]{\includegraphics[scale=0.7]{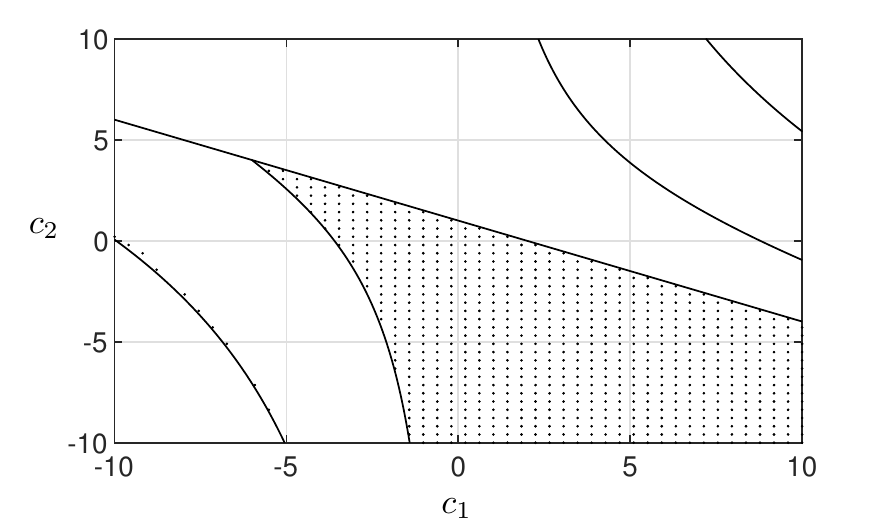}}
\subfigure[]{\includegraphics[scale=0.7]{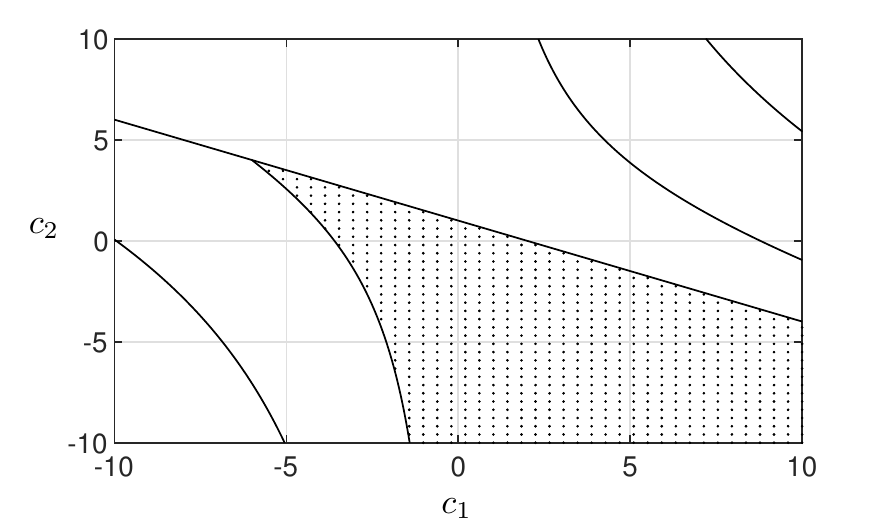}}
\caption{Evolution of the stability region of \Cref{ex:1}: (a) $r=2$, (b) $r=3$, (c) $r=4$, (d) $r=5$}
\label{FigEj1}
\end{figure}
\end{example}

\begin{example} \label{ex:2}
Let the integral equation
\[ x(t)=B\int_{-h}^{0}x(t+\theta)d\theta, \]
where $h=1$ and
\begin{equation}
B=\begin{bmatrix}
0 & 1 & 0 & 0 \\
0 & 0 & 1 & 0 \\
0 & 0 & 0 & 1 \\
-2 & -b_1 & -1 & -b_2
\end{bmatrix}.
\end{equation}
The space of parameter we are interested in is $(b_1,b_2)$. As shown in \Cref{FigEj2}, the exact stability region is reached with $r=2$.
\begin{figure}[!t]
\centering
\includegraphics[scale=0.9]{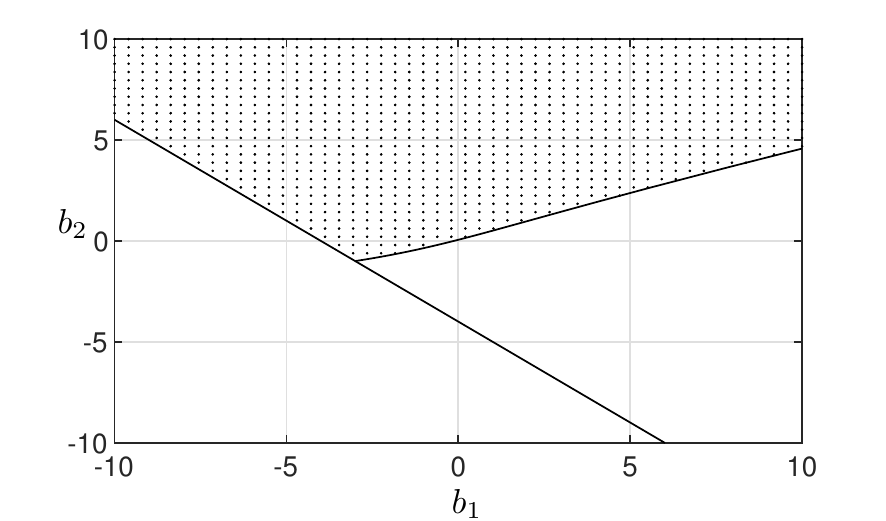}
\caption{Stability region of \Cref{ex:2} found by testing~\eqref{eqCNS1} with $r=2$}
\label{FigEj2}
\end{figure}
\end{example}

%%%%%%%%%%%%%%%%%%%%%%%%%%%%%%%%%%%%%%%%%%%%%%%%%%%%%%%%%%%%%%%%%
\section{Conclusion} \label{secConclusion}

We present stability conditions for integral delay systems that consist in checking the positivity of a $n r \times n r$ matrix which depends exclusively on the delay Lyapunov matrix. We prove that these conditions are necessary and show that for some value of $r$ they become sufficient as well. The stability criterion can be applied to determine stability regions in the space of system or design parameters. The presented examples indicate that the stability region is detected with rather small values of the parameter~$r$.

We believe that the obtained stability charts can be auxiliary in assessing the validity of parameters estimates at early stages of epidemics outbreaks, or in the choice of control parameters of spectrum assignment laws for input delay systems.

Finally, it remains to recognize how much this paper owes to the fundamental work on the Lyapunov matrix of differential systems of retarded type, neutral type, and with distributed delays  exposed in the monograph by Vladimir Kharitonov "Time-delay systems: Lyapunov functionals and matrices".

%%%%%%%%%%%%%%%%%%%%%%%%%%%%%%%%%%%%%%%%%%%%%%%%%%%%%%%%%%%%%%%%%
\section*{ORCID}
\ORCID{Reynaldo Ortiz}{0000-0001-9353-1975} \\
\ORCID{Alexey Egorov}{0000-0001-7671-2467} \\
\ORCID{Sabine Mondié}{0000-0002-0968-1899}

%%%%%%%%%%%%%%%%%%%%%%%%%%%%%%%%%%%%%%%%%%%%%%%%%%%%%%%%%%%%%%%%%
\appendix

\section{Proof of Theorem 1} \label{appQ}
%\Cref{thmQ}

Using equality~\eqref{eqK} and absolute continuity of the fundamental matrix, we can rewrite formula~\eqref{eqQ} as
\begin{equation}
Q(\xi_1,\xi_2)=I_1-I_2,
\end{equation}
where
\[ I_1=\int_{0}^{\infty} K'^T(t-\xi_1)W\left( \dfrac{d}{dt} \int_{-h}^{\xi_2}K(t-\xi_2+\theta_2)F(\theta_2) d\theta_2\right)dt, \]
\[ I_2=\int_{\xi_1}^{0}F^T(\theta_1) \int_{0}^{\infty}K'^T(t-\xi_1+\theta_1)W \left(\dfrac{d}{dt} \int_{-h}^{\xi_2} K(t-\xi_2+\theta_2) F(\theta_2) d\theta_2\right)dt d\theta_1. \]
We turn our attention to the following expression that occurs both in $I_1$ and $I_2$:
\[ R(\Delta,\xi)=\int_{0}^{\infty} K'^T(t-\Delta)W\left(\dfrac{d}{dt} \int_{-h}^{\xi}K(t-\xi+\theta) F(\theta)d\theta\right)dt. \]

\begin{lemma}
Let system~\eqref{eqSys} be exponentially stable. The following equality holds for $\Delta,\,\xi\in[-h,0)$:
\[ R(\Delta,\xi)=-\int_{-h}^{\xi} U''(\Delta-\xi+\theta)F(\theta)d\theta. \]
\end{lemma}
\begin{proof}
With function~\eqref{eqFtilde} we deduce that
\[ \int_{-h}^{\xi}K(t-\xi+\theta)F(\theta)d\theta =\int_{\xi-t}^{\xi} K(t-\xi+\theta)\F(\theta)d\theta +K_0\int_{-h}^{\xi-t}\F(\theta)d\theta \]

As $K$ is absolutely continuous on $[0,\infty)$, by~\Cref{thmLK0K0I}
\[ \dfrac{d}{dt} \int_{-h}^{\xi}K(t-\xi+\theta)F(\theta)d\theta =\F(\xi-t) +\int_{\xi-t}^{\xi}K'(t-\xi+\theta) \F(\theta)d\theta. \]

Thus,
\begin{multline}
R(\Delta,\xi)=\int_0^\infty K'^T(t-\Delta)W \left(\F(\xi-t)+\int_{\xi-t}^{\xi} K'(t-\xi+\theta) \F(\theta)d\theta\right)dt =\int_{-\infty}^{\xi} K'^T(\xi-\theta-\Delta)W\F(\theta)d\theta \\
+\int_{-\infty}^{\xi} \int_{\xi-\theta}^{\infty} K'^T(t-\Delta)WK'(t-\xi+\theta)dt\F(\theta)d\theta =\int_{-h}^{\xi} \left(K'^T(\xi-\theta-\Delta)W +\int_0^\infty K'^T(t+\xi-\theta-\Delta)WK'(t) dt\right) F(\theta)d\theta.
\end{multline}
By~\Cref{thmUprimeprime} we obtain the desired result.
\end{proof}

\Cref{thmQ} is obvious now.

%%%%%%%%%%%%%%%%%%%%%%%%%%%%%%%%%%%%%%%%%%%%%%%%%%%%%%%%%%%%%%%%%
\section{Proof of Lemma 5} \label{appLMprop1}
%\Cref{thmGDP}

\textit{Part one}. Consider the function
\begin{equation} \label{eqG}
G(\tau)=-U(\tau)+\int_{-h}^{0}F^T(\theta)U(\tau-\theta)d\theta + WR(\tau),
\end{equation}
where
\[ R(\tau)=S -\int_{0}^{\tau}K(s)ds. \]

Properties~\eqref{eqDyn},~\eqref{eqSym} allow to deduce the following equality\cite{Ortiz2019TAC}:
\begin{equation} \label{eqDnegative}
U(\tau)=\int_{-h}^{0}F^T(\theta)U(\tau-\theta)d\theta+WS-\tau WK_0,\quad \tau<0.
\end{equation}
As for $\tau<0$ function $R(\tau)=S-\tau K_0$, we conclude that $G(\tau)=0$, if $\tau<0$.

Next, we prove that for $\tau\>0$ function $G$ satisfies equation~\eqref{eqKK}. The first term in $G$ obviously does, and by~\eqref{eqDyn}, as $\tau-\theta\>0$, so does the second term. It remains to consider the last term in $G$.

Notice that by~\eqref{eqKK}
\begin{equation}
R(\tau)=S-\int_{0}^{\tau}\int_{-h}^0K(s+\theta)F(\theta)d\theta ds =S -\int_{-h}^0\int_{\theta}^{\tau+\theta} K(s)ds F(\theta)d\theta = S -\int_{-h}^0\int_0^{\tau+\theta}K(s)dsF(\theta)d\theta +K_0 \int_{-h}^0\theta F(\theta)d\theta.
\end{equation}
Finally, by~\eqref{eqS} we obtain that
\begin{equation}
R(\tau) =S +K_0\int_{-h}^0\theta F(\theta)d\theta K_0 K_0^{-1} -\int_{-h}^0\int_0^{\tau+\theta}K(s)dsF(\theta)d\theta
=\int_{-h}^0\left(S-\int_0^{\tau+\theta}K(s)ds\right)F(\theta)d\theta =\int_{-h}^0R(\tau+\theta)F(\theta)d\theta,
\end{equation}
thus, function $G$ satisfies
\begin{equation} \label{eqDynGT}
G(\tau)=\int_{-h}^{0}G(\tau+\theta)F(\theta)d\theta,
\end{equation}
for $\tau\>0$. As the initial function is zero, the unique solution of~\eqref{eqDynGT} is $G(\tau)=0$, $\tau\in\R$.

\textit{Part two}. Transpose formula~\eqref{eqAlgNeg} and substitute the symmetry property~\eqref{eqSym} to arrive at
\begin{equation}
U(-\tau)-P+\tau K_0^TWK_0 =\int_{-h}^{0}\left(U(-\tau+\theta) -P+(\tau-\theta)K_0^TWK_0\right)F(\theta)d\theta +S^TW+\int_{\tau}^0K^T(s)ds W.
\end{equation}
By~\eqref{eqK0},~\eqref{eqP} and \eqref{eqS},
\begin{equation}
P-\tau K_0^TWK_0 +\int_{-h}^{0}\left(-P+(\tau-\theta)K_0^TWK_0\right)F(\theta)d\theta +S^TW+\int_{\tau}^0K^T(s)ds W =\int_{\tau}^0K^T(s)ds W +\tau K_0^TW.
\end{equation}
Substitution of $-\tau$ instead of $\tau$ finally gives the desired equality.

%%%%%%%%%%%%%%%%%%%%%%%%%%%%%%%%%%%%%%%%%%%%%%%%%%%%%%%%%%%%%%%%%
\section{Proof of Lemma 6} \label{appLMprop2}
%\Cref{thmT1T2}

We consider the following equality:
\[ \int_{0}^{\tau_2}U'(\alpha-\xi)K(\xi)d\xi =\int_{0}^{\tau_2}U'(\alpha-s)K(s)ds. \]
Substitution of~\eqref{eqK} in the left hand side (l.h.s) and~\eqref{eqAlg} in the right hand side (r.h.s) lead to
\begin{equation}
\int_{-h}^{0}\int_{0}^{\tau_2}U'(\alpha-\xi)F(\theta)K(\xi+\theta)d\xi d\theta =\int_{-h}^{0}\int_{0}^{\tau_2} U'(\alpha-s+\theta) F(\theta)K(s) ds d\theta +\int_{0}^{\tau_2}\left(K(s-\alpha)-K_0\right)^TWK(s)ds,
\end{equation}
and the change of variable $\xi=s-\theta$ in the l.h.s. gives
\begin{equation} \label{eqLemU}
\int_{-h}^{0}\int_{\theta}^{0}U'(\alpha-s+\theta)ds F(\theta)d\theta K_0 =\int_{-h}^{0} \int_{\tau_2+\theta}^{\tau_2} U'(\alpha-s+\theta) F(\theta) K(s) ds d\theta +\int_{0}^{\tau_2}\left(K(s-\alpha)-K_0\right)^TWK(s)ds.
\end{equation}
Using the fundamental theorem of calculus, formula~\eqref{eqAlg} and properties of matrix $K$, we can transform the l.h.s. into
\begin{multline}
\int_{-h}^{0}\left(U(\alpha)-U(\alpha+\theta)\right)F(\theta)d\theta K_0 =U(\alpha)\left(\int_{-h}^{0} F(\theta)d\theta -I\right)K_0 +\int_{-\alpha}^{0}\left(K(s)-K_0\right)^T ds W K_0 \\
=U(\alpha)+\int_{\tau_2-h-\alpha}^{0}\left(K(s)-K_0\right)^T ds W K_0 -\int_{\tau_2-h}^{0} \left(K(\xi-\alpha)-K_0\right)^TWK(\xi) d\xi.
\end{multline}
Substituting this expression into the l.h.s. of~\eqref{eqLemU}, we obtain
\begin{equation}
U(\alpha)+\int_{\tau_2-h-\alpha}^{0}\left(K(s)-K_0\right)^T ds W K_0 =\int_{-h}^{0} \int_{\tau_2+\theta}^{\tau_2} U'(\alpha-s+\theta) F(\theta) K(s) ds d\theta +\int_{\tau_2-h}^{\tau_2}\left(K(s-\alpha)-K_0\right)^TWK(s)ds.
\end{equation}
If we take $\alpha=\tau_2-\tau_1$, we can notice that the integral in the l.h.s. equals zero. It remains to change the variable $s=\tau_2+\xi$ in the r.h.s and change the order of integration in the double integral to finish the proof.

%%%%%%%%%%%%%%%%%%%%%%%%%%%%%%%%%%%%%%%%%%%%%%%%%%%%%%%%%%%%%%%%%
\section{Proof of Theorem 2} \label{appDiffV0}
%\Cref{thmDv0}

In order to simplify calculations we introduce a technical result.
\begin{lemma} \label{thmQ0}
The function $Q$ is such that for any $\xi_1,\,\xi_2\in[-h,0]$
\[ Q(\xi_1,-h)=Q(-h,\xi_2)=0. \]
\end{lemma}
\begin{proof}
It follows directly from formula~\eqref{eqQ} that $Q(\xi_1,-h)=0$. Now consider 
\[ Q(-h,\xi_2) =-\int_{-h}^{\xi_2}U''(-h-\xi_2+\theta)F(\theta)d\theta +\int_{-h}^{\xi_2}\int_{-h}^{0} F^T(\theta_1)U''(-h-\theta_1-\xi_2+\theta_2) d\theta_1F(\theta_2)d\theta_2. \]
Note that $-h-\xi_2+\theta\<-h$. By formula \eqref{eqAlgNeg}, we obtain the expression
\[ U''(-h-\xi_2+\theta) =\int_{-h}^{0}F^T(\theta_1)U''(-h-\theta_1-\xi_2+\theta)d\theta_1, \]
and the result follows.
\end{proof}

We are ready to differentiate the functional $v_0$. We first apply the changes of variable $s_1=t+\xi_1$ and $s_2=t+\xi_2$ to~\eqref{eqv0Q}:
\[ \dfrac{d}{dt}v_0(x_t)=\dfrac{d}{dt}\int_{t-h}^{t}x^T(s_1)\int_{t-h}^{t} Q(s_1-t,s_2-t)x(s_2)ds_2 ds_1. \]
By the Leibnitz integral rule and \Cref{thmQ0}, the derivative of the functional $v_0$ is
\[ \dfrac{d}{dt}v_0(x_t)=x^T(t)\int_{t-h}^{t}Q(0,s_2-t)x(s_2)ds_2 +\int_{t-h}^{t}x^T(s_1)Q(s_1-t,0)ds_1 x(t) +\int_{t-h}^{t}x^T(s_1)\int_{t-h}^{t} \dfrac{d}{dt}\left[Q(s_1-t,s_2-t)\right]x(s_2)ds_2 ds_1. \]
With equality~\eqref{eqSys} we obtain
\begin{equation} \label{eqDerV0}
\dfrac{d}{dt}v_0(x_t)=\int_{t-h}^{t}\int_{t-h}^{t} x^T(s_1) \left(F^T(s_1-t)Q(0,s_2-t) +Q(s_1-t,0)F(s_2-t) +\dfrac{d}{dt}\left[Q(s_1-t,s_2-t)\right]\right)x(s_2)ds_2 ds_1.
\end{equation}

Consider the third term in brackets:
\begin{equation}
Q(s_1-t,s_2-t)=-\int_{-h}^{s_2-t}U''(s_1-s_2+\theta)F(\theta)d\theta+\int_{s_1-t}^{0}\int_{-h}^{s_2-t}F^T(\theta_1)U''(s_1-s_2-\theta_1+\theta_2)F(\theta_2)d\theta_2d\theta_1.
\end{equation}
By the Leibnitz integral rule,
\[ \dfrac{d}{dt}\left[Q(s_1-t,s_2-t)\right]=  U''(s_1-t)F(s_2-t)+F^T(s_1-t) \int_{-h}^{s_2-t} U''(t-s_2+\theta_2)F(\theta_2)d\theta_2 -\int_{s_1-t}^{0}F^T(\theta_1) U''(s_1-\theta_1-t)d\theta_1F(s_2-t). \]
With the changes of variable $s_1=t+\xi_1$ and $s_2=t+\xi_2$, we obtain the equality
\begin{multline}
F^T(\xi_1)Q(0,\xi_2) +Q(\xi_1,0)F(\xi_2) +\dfrac{d}{dt}\left[Q(s_1-t,s_2-t)\right]\Bigl|_{s_1=t+\xi_1,s_2=t+\xi_2} \\
=\left(U''(\xi_1)-\int_{\xi_1}^{0}F^T(\theta_1) U''(\xi_1-\theta_1)d\theta_1 -\int_{-h}^{0}U''(\xi_1+\theta)F(\theta) d\theta +\int_{\xi_1}^{0}\int_{-h}^{0} F^T(\theta_1)U''(\xi_1-\theta_1+\theta_2)F(\theta_2)d\theta_2 d\theta_1\right) F(\xi_2).
\end{multline}
Applying~\Cref{thmU''} to the first and to the second term, we reduce the expression to
\[ \left(-K'^T(-\xi_1)+\int_{\xi_1}^0 F^T(\theta_1)K'^T(\theta_1-\xi_1)d\theta_1\right) WF(\xi_2)=-F^T(\xi_1)WF(\xi_2). \]
The last equality holds true by~\Cref{thmKprime}. Substitution of the obtained expression into~\eqref{eqDerV0} leads to~\eqref{eqDv0}.

%%%%%%%%%%%%%%%%%%%%%%%%%%%%%%%%%%%%%%%%%%%%%%%%%%%%%%%%%%%%%%%%%
\section{Proof of Lemma 7} \label{appCIZ}
%\Cref{thmCIZ}

By~\eqref{eqK0}, substitution of $\vph_1$, $\vph_2$ into the bilinear functional~\eqref{eqZ} gives
\begin{equation} \label{eqZT12}
z(\vph_1,\vph_2)=\gamma_1^T \left(T(\tau_1,\tau_2)-T(\tau_1,0)-T(0,\tau_2)+T(0,0)\right)\gamma_2,
\end{equation}
where
\begin{equation}
T(\tau_1,\tau_2)=\int_{-h}^{0}\int_{-h}^{0}K^T(\tau_1+\xi_1)Q(\xi_1,\xi_2) K(\tau_2+\xi_2)d\xi_2 d\xi_1 +\int_{-h}^{0}K^T(\tau_1+\xi)WK(\tau_2+\xi)d\xi.
\end{equation}
With the function
\[ P(s)=\int_{-h}^0 \int_{-h}^{\xi} U'(-s-h-\xi+\theta)F(\theta)d\theta K(\tau_2+\xi)d\xi \]
the first term of $T$ is written as
\begin{equation}
J=\int_{-h}^{0}\int_{-h}^{0}K^T(\tau_1+\xi_1)Q(\xi_1,\xi_2) K(\tau_2+\xi_2)d\xi_2 d\xi_1 =\int_{-h}^{0}K^T(\tau_1-s-h) P'(s) ds -\int_{-h}^{0}\int_{-h-s}^{0}  K^T(\tau_1-s-h) F^T(\theta) P'(s+\theta) d\theta ds.
\end{equation}
Changing the order of integration in the second term of the r.h.s and using the change of variable $s=\xi-\theta$, yields
\[ J=\int_{-h}^{0}K^T(\tau_1-s-h)P'(s)ds -\int_{-h}^{0}\int_{-h}^{\theta}  K^T(\tau_1-\xi+\theta-h) F^T(\theta) P'(\xi) d\xi d\theta. \]
With another change of the order of integration by definition of the fundamental matrix we obtain
\begin{multline}
J=\int_{-h}^{\tau_1-h} \left(K(\tau_1-\xi-h) -\int_{\xi}^{0} F(\theta) K(\tau_1-\xi-h+\theta)d\theta\right)^T P'(\xi) d\xi \\
+\int_{\tau_1-h}^{0} \left(K(\tau_1-\xi-h) -\int_{\xi}^{0} F(\theta) K(\tau_1-\xi-h+\theta)d\theta\right)^T P'(\xi) d\xi =K_0^T\int_{-h}^{\tau_1-h} \int_{-h}^{\xi} F^T(\theta)d\theta P'(\xi) d\xi \\
+K_0^T\int_{\tau_1-h}^{0} \left(I -\int_{-h}^{0} F^T(\theta)d\theta +\int_{-h}^{\xi} F^T(\theta)d\theta \right) P'(\xi) d\xi =K_0^T\int_{-h}^0 \int_{-h}^{\xi} F^T(\theta)d\theta P'(\xi) d\xi -\int_{\tau_1-h}^{0} P'(\xi) d\xi.
\end{multline}
By integration by parts, we get
\begin{equation}
J=K_0^T\int_{-h}^0F^T(\theta)d\theta P(0)-K_0^T\int_{-h}^0 F^T(\xi)P(\xi)d\xi-P(0)+P(\tau_1-h) =K_0^TP(0)-K_0^T\int_{-h}^0 F^T(s)P(s)ds+P(\tau_1-h).
\end{equation}
Consider now the difference
\begin{equation}
P(0)-\int_{-h}^0 F^T(s)P(s)ds =\int_{-h}^0 \int_{-h}^{\xi} \left(U'(-h-\xi+\theta) -\int_{-h}^0 F^T(s)U'(-s-h-\xi+\theta)ds\right)F(\theta)d\theta K(\tau_2+\xi)d\xi.
\end{equation}
As $-h-\xi+\theta\<0$, by~\eqref{eqAlgNegDer}
\begin{multline}
P(0)-\int_{-h}^0 F^T(s)P(s)ds =-\int_{-h}^0\int_{-h}^{\xi}WK_0F(\theta)d\theta K(\tau_2+\xi)d\xi =-WK_0\int_{-h}^0 F(\theta)\int_{\tau_2+\theta}^{\tau_2}K(s)ds d\theta\\
=-WK_0\int_{-h}^0 F(\theta)\int_{\theta}^{0}K(s)ds d\theta -WK_0\int_{-h}^0 F(\theta)d\theta \int_{0}^{\tau_2} K(s)ds +WK_0\int_0^{\tau_2}\int_{-h}^0 F(\theta)K(s+\theta)d\theta ds\\
= WS -WK_0\left(\int_{-h}^0 F(\theta)d\theta-I\right)\int_{0}^{\tau_2}K(s)ds 
=WS-W\int_{-\tau_2}^0 K(\tau_2+\theta)d\theta.
\end{multline}
By~\Cref{thmT1T2},
\[ P(\tau_1-h)=U(\tau_2-\tau_1) -\int_{-h}^0 \left(K(\tau_1+\xi)-K_0\right)^T W K(\tau_2+\xi)d\xi. \]
Substituting of the obtained equalities into $T$ gives
\begin{multline}
T(\tau_1,\tau_2)=K_0^TWS -K_0^TW\int_{-\tau_2}^0 K(\tau_2+\theta)d\theta +U(\tau_2-\tau_1) +K_0^T W \int_{-h}^0 K(\tau_2+\xi)d\xi \\
= U(\tau_2-\tau_1) + K_0^TWS + (h-\tau_2)K_0^TWK_0.
\end{multline}
In view of the above,~\eqref{eqZT12} can be written as
\begin{multline}
z(\vph_1,\vph_2)=\gamma_1^T \left\{ U(\tau_2-\tau_1)+K_0^TWS+(h-\tau_2)K_0^TWK_0 
- U(-\tau_1)-K_0^TWS-hK_0^TWK_0 \right. \\
\left. -U(\tau_2)-K_0^TWS-(h-\tau_2)K_0^TWK_0 +U(0)+K_0^TWS+hK_0^TWK_0\right\} \gamma_2,
\end{multline}
and the result follows.

%%%%%%%%%%%%%%%%%%%%%%%%%%%%%%%%%%%%%%%%%%%%%%%%%%%%%%%%%%%%%%%%%
\section{Proof of Theorem 4} \label{appUnsTh}
%\Cref{thmUnsTh}

Start with a technical result.

\begin{lemma} \label{thmGB}
Let $x_1,\ldots,x_r\in\R$, $f,c_2,\ldots,c_r\>0$, and $c_2+\ldots+c_r=c$. If
\begin{equation} \label{eqGB}
x_k\<c_k+f\sum_{i=1}^{k-1} x_i,\;\;k=2,\ldots,r,
\end{equation}
then
\[ \sum_{k=1}^r x_k\< (c+x_1)(1+f)^{r-1}. \]
\end{lemma}
\begin{proof}
With the substitution $y_k=x_k-c_k$, $k=1,\ldots,r$, with $c_1=0$, inequalities~\eqref{eqGB} take the form
\[ y_k\<f \sum_{i=1}^{k-1} c_i +f\sum_{i=1}^{k-1} y_i \<fc+f\sum_{i=1}^{k-1} y_i,\;\;k=2,\ldots,r. \]
By mathematical induction $y_k\<f(c+y_1)(1+f)^{k-2}$, $k=2,\ldots,r$.
To finish the proof it remains to make an inverse substitution and summarize estimates for $x_1,\ldots,x_r$.
\end{proof}

The proof of \Cref{thmUnsTh} is based on the following two lemmas.

\begin{lemma} \label{thmApprox}
For every function $\vph\in\PC_h$ and every number $\delta>0$ there exists a function $\psi$ of the form~\eqref{eqSpec} with equidistant points~\eqref{eqEquidTau}, such that
\begin{equation} \label{eqApprox}
\int_{-h}^0 \|\vph(\theta)-\psi(\theta)\|\,d\theta <\delta.
\end{equation}
\end{lemma}
\begin{proof}
Given a function $\vph$ and a number $\delta>0$, we seek a natural number $r$ and vectors $\gamma_i$, $i=1,\ldots,r$, such that inequality~\eqref{eqApprox} holds true for
\[ \psi(\theta)=\sum_{i=1}^{r} \left(K(\tau_i+\theta)-K_0\right) \gamma_i, \]
where $\tau_i=ih/r$, $i=1,\ldots,r$. We uniquely define vectors $\gamma_1,\ldots,\gamma_r$ from the equalities
\[ \vph(-\tau_k)=\psi(-\tau_k), \quad k=1,\ldots,r, \]
which can be rewritten, based on the definition of the fundamental matrix, in an explicit form:
\begin{gather}
\gamma_r=\vph(-h), \\
\gamma_k=\vph(-\tau_k)-\sum_{i=k+1}^{r} \left(K(\tau_i-\tau_k)-K_0\right)\gamma_i,\quad k=1,\ldots,r-1.
\end{gather}
It remains to deduce an estimate of the distance between $\vph$ and $\psi$, which tends to zero as $r$ tends to infinity.

As the fundamental matrix is absolutely continuous on $[0,h]$, it is also Lipschitz, i.\:e., there exists a constant $L$ such that
\[ \|K(t_1)-K(t_2)\|\< L|t_1-t_2|, \quad t_1,t_2\in[0,h]. \]

Consider first $\theta\in[-\tau_{r},-\tau_{r-1})$:
\[ \|\vph(\theta)-\psi(\theta)\| =\|\vph(\theta)-\vph(-\tau_r) +\psi(-\tau_r) -\psi(\theta)\| \< \|\vph(\theta) -\vph(-\tau_r)\| +\left\|K(\tau_r+\theta)-K(0)\right\| \|\gamma_r\| \<\bigvee_{-\tau_r}^{-\tau_{r-1}}\vph +(\tau_r+\theta)L\|\gamma_r\|, \]
where the first term on the right hand side denotes the variation of function $\vph$ on the segment $[-\tau_r,-\tau_{r-1}]$.

Similarly, one can show that for $\theta\in[-\tau_{k},-\tau_{k-1})$, $k\in\{1,\ldots,r\}$ ($\tau_0=0$),
\[ \psi(\theta)=\sum_{i=k}^{r} \left(K(\tau_i+\theta)-K_0\right) \gamma_i, \]
and therefore,
\[ \left\|\vph(\theta)-\psi(\theta)\right\| \<\|\vph(\theta)-\vph(-\tau_k)\|+\sum_{i=k}^r\left\|K(\tau_i+\theta)-K(\tau_i-\tau_k)\right\|\|\gamma_i\| \<\bigvee_{-\tau_k}^{-\tau_{k-1}}\vph +(\tau_k+\theta)L\sum_{i=k}^r\|\gamma_i\|. \]
In this research, we are not interested in the high accuracy of the upper boundaries, so we can roughly estimate this expression from above as follows:
\begin{equation} \label{eqIneqApprox}
\left\|\vph(\theta)-\psi(\theta)\right\|\<\bigvee_{-\tau_k}^{-\tau_{k-1}}\vph +\dfrac{h}{r}L\sum_{i=1}^r\|\gamma_i\|.
\end{equation}

By definition of $\psi$, one can deduce for $k\in\{1,\ldots,r-1\}$ that $\gamma_k=\psi(-\tau_k)-\psi(-\tau_k-0)$, where $\psi(-\tau_k-0)$ means the left limit of $\psi$ at point $-\tau_k$. This implies that
\[ \|\gamma_k\|\< \|\vph(-\tau_k)-\vph(-\tau_{k+1})\| +\|\psi(-\tau_{k+1})-\psi(-\tau_k-0)\| \<\bigvee_{-\tau_{k+1}}^{-\tau_{k}}\vph +\dfrac{h}{r}L\sum_{i=k+1}^r\|\gamma_i\|. \]

Now we can apply \Cref{thmGB} with
\[ f=\dfrac{h}{r}L, \quad x_k=\|\gamma_{r-k+1}\|, \;\;k=1,\ldots,r, \quad c_k=\bigvee_{-\tau_{r-k+2}}^{-\tau_{r-k+1}}\vph, \;\;k=2,\ldots,r,  \] 
to obtain that
\[ \sum_{i=1}^r\|\gamma_i\|\< (c+\|\gamma_r\|)\left(1+\dfrac{h}{r}L\right)^{r-1} \<\left(\bigvee_{-h}^{0}\vph+\|\vph(-h)\|\right)\left(1+\dfrac{h}{r}L\right)^{r-1} \<\left(\bigvee_{-h}^{0}\vph+\|\vph(-h)\|\right)\left(1+\dfrac{h}{r}L\right)^{-1}e^{Lh}. \]

Combining the obtained inequality with~\eqref{eqIneqApprox}, we obtain the desired result:
\[ \int_{-h}^{0} \|\vph(\theta) -\psi(\theta)\|\,d\theta \<\sum_{k=1}^r \int_{-\tau_k}^{-\tau_{k-1}} \|\vph(\theta) -\psi(\theta)\|\,d\theta \<\dfrac{h}{r}\bigvee_{-h}^{0}\vph +he^{Lh}\left(\dfrac{r}{Lh}+1\right)^{-1} \xrightarrow[r\rightarrow \infty]{} 0.\tag*{\qedhere} \]
\end{proof}

\begin{lemma}
If system~\eqref{eqSys} is unstable and satisfies \Cref{assEig}, then for every $c>0$ there exists a function $\vph\in\PC_h$, such that
\[ v_1(\vph)\<-c. \]
\end{lemma}
\begin{proof}
As functional $v_1$ is homogeneous, it remains to show that there exists a function $\vph\in\PC_h$, such that $v_1(\vph)<0$.

As system~\eqref{eqSys} is unstable and does not have pure imaginary eigenvalues, it has an eigenvalue with a positive real part. To prove this, introduce system with a term described by the Riemann-Stieltjes integral
\begin{equation} \label{eqDiffSys}
\dfrac{d}{dt}x(t)=F(0)x(t)-F(-h)x(t-h) -\int_{-h}^0 dF(\theta)x(t+\theta) -\left(x(t)-\int_{-h}^0 F(\theta) x(t+\theta)d\theta\right),
\end{equation}
which can be obtained from~\eqref{eqSys} by applying the operator $\dfrac{d}{dt}+\mathrm{id}$. It is clear that any absolutely continuous solution of~\eqref{eqSys} is also a solution of system~\eqref{eqDiffSys}. But there is one problem. There exist solutions of~\eqref{eqSys} that are not absolutely continuous on $[0,\infty)$.

However, one can show that all solutions of system~\eqref{eqSys} are absolutely continuous on $[h,\infty)$, i.\:e., integral systems have the property of smoothing the solutions. It can can be shown in the same way, like for the fundamental matrix in \Cref{thmKprime}, taking into account that on $[0,h]$ any solution is continuous. Thus, we can claim that system~\eqref{eqDiffSys} is unstable, as~\eqref{eqSys} is.

The connection between eigenvalues of system~\eqref{eqDiffSys} and stability is well established, in contrast to integral systems. Compute the characteristic matrix of system~\eqref{eqDiffSys}:
\[ \widetilde H(s)=sI -F(0)+I+F(-h)e^{-sh}-\int_{-h}^0 e^{s\theta}F(\theta)d\theta +\int_{-h}^0e^{s\theta} dF(\theta). \]
Integration by parts in the last term leads to the equality
\[ \widetilde H(s)=sI+I-\int_{-h}^0 e^{s\theta}F(\theta)d\theta -s\int_{-h}^0e^{s\theta} F(\theta)d\theta =(s+1)H(s), \]
where $H$ is the characteristic matrix for system~\eqref{eqSys}. Thus, system~\eqref{eqDiffSys} has the same eigenvalues, like system~\eqref{eqSys}, but additionally has multiple eigenvalue $-1$.

System~\eqref{eqDiffSys} is unstable and, as we see now, does not have pure imaginary eigenvalues by \Cref{assEig}. Thus, it has an eigenvalue $s_0=\alpha+j\beta$ with positive real part ($\alpha>0$), and $s_0$ is also an eigenvalue of~\eqref{eqSys}. Let $C=C_1+jC_2$ be a corresponding eigenvector. Then
\[ \tilde x(t)=e^{\alpha t}\left(C_1\cos(\beta t) -C_2\sin(\beta t)\right),\;\;t\>-h, \]
is a nontrivial solution of~\eqref{eqSys}.

If $\beta\neq 0$ take $\tau=2\pi/|\beta|$, while if $\beta=0$ take $\tau=1$. It follows that
\[ \tilde x(\tau+\theta)=e^{\alpha\tau}\tilde x(\theta), \quad \theta\in[-h,0). \]
As $v_1$ is a quadratic functional,
\[ v_1(\tilde x_\tau) =e^{2\alpha\tau}v_1(\tilde x_0). \]
Integrating~\eqref{eqDv1} along the solution $\tilde x$ from $0$ to $\tau$, we obtain
\begin{equation} \label{eqPfUnsTh1}
-\int_{-h}^{\tau-h}\tilde x^T(t)W\tilde x(t)dt =v_1(\tilde x_\tau) -v_1(\tilde x_0) =\left(e^{2\alpha\tau}-1\right) v_1(\tilde x_0).
\end{equation}
It is obvious now that $v_1(\tilde x_0)<0$.
\end{proof}

To finish the proof of \Cref{thmUnsTh}, one just need to combine the two lemmas presented above with the continuity of functional $v_1$ (see, \Cref{thmV1prop}).

\bibliography{wileyNJD-AMA}

\begin{thebibliography}{10}
\providecommand \doibase [0]{http://dx.doi.org/}%

\bibitem{Euler1760}
Euler L. Recherches g\'en\'erales sur la mortalit\'e et la multiplication du
  genre humain. {\it M\'em.Acad. R. Sci. Belles Lett.} 1760\string; 16\string:
  144-164.

\bibitem{Lotka1907}
Lotka AJ. Relation between birth rates and death rates. {\it Science}
  1907\string; 26\string: 21-22.

\bibitem{KermackMcKendrick1927}
Kermack WO, McKendrick AG. A Contribution to the Mathematical Theory of
  Epidemics. {\it Proceedings of the Royal Society of London. Series A,
  Containing Papers of a Mathematical and Physical Character} 1927\string;
  115(772)\string: 700-721.

\bibitem{Dushoff2018Equivalence}
Champredon D, Dushoff J, Earn DJD. Equivalence of the Erlang-Distributed SEIR
  Epidemic Model and the Renewal Equation. {\it SIAM Journal on Applied
  Mathematics} 2018\string; 78(6)\string: 3258-3278.

\bibitem{DushoffInfering2020}
Woo PS, David C, Jonathan D. Inferring generation-interval distributions from
  contact-tracing data. {\it J. R. Soc. Interface} 2020\string; 17\string:
  1-12.

\bibitem{BellmanCooke}
Bellman R, Cooke KL. {\it Differential-difference equations}.
\newblock New York: Academic Press .
\newblock 1963.

\bibitem{Cooke1976}
Cooke KL, Kaplan JL. A Periodicity Threshold Theorem for Epidemic and
  Population Growth. {\it Mathematical Biosciences} 1976\string; 31\string:
  87-104.

\bibitem{Torrejon1976}
Torrejón R. A periodicity Threshold Theorem for Epidemics and Population
  Growth. {\it Mathematical Biosciences} 1976\string; 31\string: 87-104.

\bibitem{London1973}
London WP, Yorke JA. Recurrent outbreaks of measles chickenpox and mumps. {\it
  American Journal of Epidemiology} 1973\string; 98(6)\string: 453-468.

\bibitem{ArinoDriessche2006}
Arino J, Driessche v.~dP. Time delays in epidemic models: modeling and
  numerical considerations. In:  Arino O, Hbid M, Dads EA. \kern-2pt, eds. {\it
  Delay Differential Equations and Applications.\;}Springer; 2006; Netherlands,
  Dordrecht\string: 539 - 578.

\bibitem{Breda2012}
Breda D, Diekmann O, Graaf dWF, Pugliese A, Vermiglio R. On the formulation of
  epidemic models (an appraisal of Kermack and McKendrick). {\it Journal of
  Biological Dynamics} 2012\string; 6(sup2)\string: 103-117.

\bibitem{ManitusOlbrot1979}
Manitus AZ, Olbrot AW. Finite spectrum assignment problem for systems with
  delays. {\it IEEE Trans. on Automatic Control} 1979\string; 24(4)\string:
  541-553.

\bibitem{Michiels2004}
Michiels W, Mondi\'e S, Roose D, Dambrine M. The Effect of Approximating
  Distributed Delay Control Laws on Stability. In:  Niculescu SI, Gu K.
  \kern-2pt, eds. {\it Advances in Time-Delay Systems}. 38 of {\it Lecture
  Notes in Computational Science and Engineering}. Springer; 2004; Berlin,
  Heidelberg\string: 207-222.

\bibitem{M-Aguilar2010}
Melchor-Aguilar D. On stability of integral delay systems. {\it Applied
  Mathematics and Computation} 2010\string; 217(7)\string: 3578-3584.

\bibitem{MondieMelchor2012}
Mondi\'e S, Melchor-Aguilar D. Exponential Stability of Integral Delay Systems
  With a Class of Analytic Kernels. {\it IEEE Trans. on Automatic Control}
  2012\string; 57(2)\string: 484-489.

\bibitem{Damak2014}
Damak S, Loreto MD, Lombardi W, Andrieu V. Exponential {L}2-stability for a
  class of linear systems governed by continuous-time difference equations.
  {\it Automatica} 2014\string; 50(12)\string: 3299-3303.

\bibitem{Damak2015}
Damak S, Di~Loreto M, Mondié S. Stability of linear continuous-time difference
  equations with distributed delay: Constructive exponential estimates. {\it
  Int. J. of Robust and Nonlinear Control} 2015\string; 25(17)\string:
  3195-3209.

\bibitem{Melchor2016}
Melchor-Aguilar D. New results on robust exponential stability of integral
  delay systems. {\it Int. J. of Systems Science} 2016\string; 47(8)\string:
  1905-1916.

\bibitem{Kharitonov2013}
Kharitonov VL. {\it Time-delay systems: {L}yapunov functionals and matrices}.
\newblock Basel: Birkh\"auser .
\newblock 2013.

\bibitem{Egorov2014}
Egorov AV. A new necessary and sufficient stability condition for linear
  time-delay systems. {\it IFAC Proceedings Volumes} 2014\string; 47(3)\string:
  11018 - 11023.
\newblock 19th IFAC World Congress.

\bibitem{EgorovCuvas2017}
Egorov AV, Cuvas C, Mondié S. Necessary and sufficient stability conditions
  for linear systems with pointwise and distributed delays. {\it Automatica}
  2017\string; 80\string: 218 - 224.

\bibitem{Gomez2018}
Gomez MA, Egorov AV, Mondié S. A new stability criterion for neutral-type
  systems with one delay. {\it IFAC-PapersOnLine} 2018\string; 51(14)\string:
  177 - 182.
\newblock 14th IFAC Workshop on Time Delay Systems TDS 2018.

\bibitem{DMelchor2010}
Melchor-Aguilar D, Kharitonov V, Lozano R. Stability conditions for integral
  delay systems. {\it Int. J. of Robust and Nonlinear Control} 2010\string;
  20(1)\string: 1-15.

\bibitem{Ortiz2019TAC}
{Ortiz} R, {Del Valle} S, {Egorov} A, {Mondié} S. Necessary stability
  conditions for integral delay systems. {\it IEEE Transactions on Automatic
  Control} 2019.

\bibitem{Ortiz2019}
Ortiz R, Mondié S. On the Lyapunov Matrix for Integral Delay Systems with a
  Class of General Kernel. {\it IFAC-PapersOnLine} 2019\string; 52(18)\string:
  91 - 96.
\newblock 15th IFAC Workshop on Time Delay Systems TDS 2019.

\bibitem{Arismendi2019}
Arismendi-Valle H, Melchor-Aguilar D. On the Lyapunov matrices for integral
  delay systems. {\it International Journal of Systems Science} 2019\string;
  50(6)\string: 1190-1201.

\bibitem{Arismendi2017}
Arismendi-Valle H, Melchor-Aguilar D. Numerical Computation of Lyapunov
  Matrices for Integral Delay Systems. {\it IFAC-PapersOnLine} 2017\string;
  50(1)\string: 13342 - 13347.
\newblock 20th IFAC World Congress.

\bibitem{Ortiz2018}
Ortiz R, Mondi\'e S, Valle SD, Egorov AV. Construction of Delay {L}yapunov
  Matrix for Integral Delay Systems. {\it Proceedings of the 57th IEEE
  Conference on Decision and Control} 2018\string: 5439-5444.
\newblock Miami Beach, Florida, USA.

\bibitem{MedvZh2013}
Medvedeva I, Zhabko A. Constructive method of linear systems with delay
  stability analysis. {\it IFAC Proceedings Volumes} 2013\string; 46(3)\string:
  1-6.
\newblock 11th Workshop on Time-Delay Systems.

\bibitem{Rocha2017}
Campos ER, Mondié S, Loreto MD. Necessary Stability Conditions for Linear
  Difference Equations in Continuous Time. {\it IEEE Trans. on Automatic
  Control} 2018\string; 63(12)\string: 4405-4412.

\bibitem{EgorovMondie2014}
Egorov AV, Mondié S. Necessary stability conditions for linear delay systems.
  {\it Automatica} 2014\string; 50(12)\string: 3204-3208.

\bibitem{Gomez2017}
Gomez MA, Egorov AV, Mondi\'e S. Necessary Stability Conditions for Neutral
  Type Systems With a Single Delay. {\it IEEE Trans. on Automatic Control}
  2017\string; 62(9)\string: 4691-4697.

\end{thebibliography}

% \clearpage

% \section*{Author Biography}

% \begin{biography}{\includegraphics[width=66pt,height=86pt,draft]{empty}}{\textbf{Author Name.} This is sample author biography text.}
% \end{biography}

\end{document}